\documentclass[a4paper,reqno,12pt]{amsart}
\usepackage[T2A]{fontenc}
\usepackage[cp1251]{inputenc}
\usepackage[russian,english]{babel}
\usepackage{amsfonts, amssymb, amsmath}

\usepackage{geometry}
\theoremstyle{plain}

\newtheorem{corollary}{Corollary}

\newtheorem*{theorem}{Theorem}

\theoremstyle{remark}
\newtheorem{remark}{Remark}

\theoremstyle{definition}
\newtheorem{example}{Example}

\DeclareMathOperator{\artanh}{artanh}

\newcommand{\R}{\mathbb{R}}
\newcommand{\Z}{\mathbb{Z}}

\begin{document}

\title{Factorization Problem with Intersection}
\thanks{}
\author{R. A. Atnagulova, O. V. Sokolova}
\address[R. A. Atnagulova]{M. Akmullah Bashkir State Pedagogical University}
\email{rushania2009@yandex.ru}
\address[O. V. Sokolova]{Lomonosov Moscow State University, Department of Mechanics and Mathematics}
\email{Olga.Efimovskaya@gmail.com}

\begin{abstract}
We propose a generalization of the factorization method to the case when $\mathcal{G}$ is a finite dimensional Lie algebra such that $\mathcal{G}=\mathcal{G}_0\oplus M \oplus  N$ (direct sum of vector spaces), where $\mathcal{G}_0$ is a subalgebra in $\mathcal{G}$, $M, N$ are $\mathcal{G}_0$-modules, and $\mathcal{G}_0 +M$, $\mathcal{G}_0 +N$ are subalgebras in  $\mathcal{G}$. In particular, we consider the case when $\mathcal{G}$ is a $\Z$-graded Lie algebra.
Using this generalization, we construct some top-like systems related to the algebra $so(3,1)$.
According to the general scheme, these systems can be reduced to linear systems with variable coefficients.
For the top-like systems polynomial first integrals and infinitesimal symmetries are found.


\end{abstract}
\maketitle
\thispagestyle{empty}


\section{Introduction}

The classical factorization method \cite{kon, TyanSh, golod, golsok} (another name is the Adler-Konstant-Symes scheme) allows to integrate the ODE system of the following special form:  
\begin{equation}\label{volchok1}
U_{t}=[U_{+},U], \quad U(0)=U_{0}.
\end{equation}
Here $U({t})$ is a  $\mathcal{G}$-valued function, where $\mathcal{G}$ is a Lie algebra, which is a direct sum of vector spaces  
$\mathcal{G}_{+}$ and $\mathcal{G}_{-}$ being subalgebras in  $\mathcal{G}$. By $U_{+}$ we denote the projection of $U$
onto $\mathcal{G}_{+}$. 
For simplicity we assume that   $\mathcal{G}$ is embedded into the matrix algebra.

The solution of problem \eqref{volchok1} is defined by the formula  
\begin{equation}\label{volchok2}
U(t)=A(t)U_{0}A^{-1}(t).
\end{equation}
The matrix $A(t)$ in \eqref{volchok2} is defined as a solution of the factorization problem  
\begin{equation}\label{volchok3}
A^{-1}B=\exp(-U_{0}t),\quad A \in G_{+},\quad B \in G_{-},
\end{equation}
where $G_{+}$ and $G_{-}$  are the Lie groups of the algebras $\mathcal{G}_{+}$ and
$\mathcal{G}_{-}$, correspondingly. If $\mathcal{G}_{-}$ is an ideal, 
then the factorization problem can be solved explicitly:  $A=\exp ((U_{0})_{+}t)$, $B=A \exp(-U_{0}t)$. In the case when
 $G_{+}$ and $G_{-}$ are algebraic groups, the conditions  
$A \in G_{+}$ and $A \exp(-U_{0}t) \in G_{-}$ yield a system of algebraic equations.  From this system the matrix $A$ for small $t$ is uniquely defined.

It was shown in \cite{golsok} by Golubchik and Sokolov that the classical factorization problem can be reduced to a system of linear ordinary differential equations with variable coefficients. 

Moreover, in \cite{golsok} the factorization method was generalized to the case 
\begin{equation}\label{volchok7}
\mathcal{G}=V_{1}\oplus V_{2},
\end{equation}
where $V_{1}$, $V_{2}$ are vector subspaces  in $\mathcal{G}_{+}$ and
$\mathcal{G}_{-}$, correspondingly. It was shown that if  
\begin{equation}\label{volchok8}
[\mathcal{G}_{+}\cap\mathcal{G}_{-},\, V_{i}]\subset V_{i}, \qquad i=1,2,
\end{equation}
then the integration of equation \eqref{volchok1}, where <<+>> means the projection onto 
$V_{1}$ parallel to $V_{2}$, can be also reduced to solving of a system of linear ODEs with variable coefficients. 

Recall the Golubchik-Sokolov construction.  
In \cite{golsok} the authors have considered a factorization problem with {\it right} logariphmic derivatives of  $A$ and $B$:
\begin{equation}\label{volchok9}
A^{-1}B=Z(t), \quad A_{t}A^{-1}\in V_{1}, \quad B_{t}B^{-1}\in
V_{2}, \quad Z(0)=A(0)=B(0)=E,
\end{equation}
where $V_i$ are vector spaces satifying \eqref{volchok7}.
Let $U(t)=A(t)q(t)A^{-1}(t)$, where $q(t)=-Z_{t}Z^{-1}$. It is easy to see that this function satisfies the following equation:
\begin{equation}\label{volchok12}
U_{t}=[U_{+},U]+Aq_{t}A^{-1}.
\end{equation}
Notice that if  $Z(t)=\exp(-U_{0}t)$ then $q(t)=U_{0}$, $q_{t}=0$ and equation   \eqref{volchok12} 
coincides with  \eqref{volchok1}.
Thus, any solution of this factorization problem yields a solution of non-linear ODE \eqref{volchok1}.

Consider now the following factorization problem: 
\begin{equation}\label{volchok13}
\alpha^{-1}\beta=Z(t), \quad \alpha^{-1}\alpha_{t}\in V_{1}, \quad
\beta^{-1}\beta_{t}\in V_{2}, \quad Z(0)=\alpha(0)=\beta(0)=E.
\end{equation}
Here the both logariphmic derivatives are {\it left}. In general, this problem does not related to equations of the \eqref{volchok1} type. But it can be reduced to a linear equation with variable coefficients. Namely, in \cite{golsok} have been introduced the linear mapping  $L(t):V_{1}\rightarrow V_{1}$, defined by 
$$L(t)(v)=(Z^{-1}(t)vZ(t))_{+}.$$
Since  $L(0)$ is the identity mapping, $L(t)$ is invertible for small $t$. It can be proved that the solution $\alpha$ of the linear equation 
\begin{equation}\label{volchok14}
\alpha_{t}=-\alpha L^{-1}(t)((Z^{-1}Z_{t})_{+}), \qquad \alpha(0)=E
\end{equation}
and the function  
\begin{equation}\label{volchok10}
\beta=\alpha Z(t),
\end{equation}
define the unique solution of the factorization problem  \eqref{volchok13}.

The last step of the construction allows us 
to establish a link between the two factorization problems 
under additional condition \eqref{volchok8}.
Namely, let $V_{1} \subset \mathcal{G}_{+}$, $V_{2} \subset \mathcal{G}_{-}$, where $\mathcal{G}_{+}$ and 
$\mathcal{G}_{-}$ are Lie subalgebras of the algebra $\mathcal{G}$ such that  $\mathcal{G}_{+}\cap\mathcal{G}_{-}=\mathcal{G}_{0}\not = {\{0\}}$. Then the solutions of the factorization problems  
\eqref{volchok9} and \eqref{volchok13} satisfy the same factorization problem:
\begin{equation}\label{volchok15}
A^{-1}B=Z(t), \quad A \in G_{+}, \quad B\in G_{-},\quad A(0)=B(0)=E,
\end{equation}
where $G_{+}$ and $G_{-}$ are the Lie groups of the algebras  $\mathcal{G}_{+}$ and 
$\mathcal{G}_{-}$. 
Since
$\mathcal{G}_{0}\neq\{0\}$, the solution of problem \eqref{volchok15} is not unique. 
Let $\alpha, \beta$ be the solution of the factorization problem \eqref{volchok13}. Because it also satisfies  \eqref{volchok15}, all solutions of \eqref{volchok15} are  given by  
\begin{equation}\label{volchok16}
A=H \alpha, \quad B=H\beta, \quad H(0)=E,
\end{equation}
where $H$ is an arbitrary element of the Lie group 
$G_{0}$ of the Lie algebra $\mathcal{G}_{0}$. 

In \cite{golsok} it was shown that  $A,B$ satisfy the factorization problem  \eqref{volchok9} if
$H$ is the solution of the following linear equation:
\begin{equation}\label{volchok17}
H_{t}=-H((\alpha_{t}\alpha^{-1})_{-}+(\beta_{t}\beta^{-1})_{+}), \quad
H(0)=E.
\end{equation}

Thus we solve equation \eqref{volchok14} for $\alpha$, find function $\beta$ from  \eqref{volchok10} and solving equation  \eqref{volchok17}, we determine the solution  $A,B$ of \eqref{volchok9} from \eqref{volchok16}. Now we can find the solution of equation   \eqref{volchok12} by formula (\ref{volchok2}).

This paper is organized as follows. In Section 2 we generalize the factorization method to the case 
when  $\mathcal{G}$ is a finite dimensional Lie algebra such that   
$\mathcal{G}=\mathcal{G}_0\oplus M \oplus  N$ (direct sum of vector spaces), 
where $\mathcal{G}_0$ is a subalgebra in $\mathcal{G}$,  $M, N$ are $\mathcal{G}_0$-modules, 
and  $\mathcal{G}_0 +M$, $\mathcal{G}_0 +N$ are subalgebras in  $\mathcal{G}$.  

The graded Lie algebras provide an important special class for which our generalization is applicable. 
In Section 3 we construct some top-like systems related to the Lie algebras  $sl(2)$ and  $so(3,1)$. 
According to the general scheme, these systems can be reduced to linear systems of ODEs with variable coefficients.  
For all these systems we find polynomial first integrals and infinitesimal symmetries and show that the systems can be integrated in quadratures by the Lie algorithm. But they do not satisfy the Painlev\'e test because of movable branching points.

\section{General construction}

Suppose we have a decomposition  
$$\mathcal{G}=\mathcal{G}_0\oplus M \oplus  N$$ 
of a finite dimensional Lie algebra  $\mathcal{G}$ over $\R$ into a vector space direct sum of a Lie subalgebra   $\mathcal{G}_0$ and two vector spaces  $M, N$ such that 
\begin{itemize}
\item $M, N$ are $\mathcal{G}_0$-modules;  
\item $\mathcal{G}_0 +M$, $\mathcal{G}_0 +N$ are subalgebras in $\mathcal{G}$.
\end{itemize}

\begin{theorem}
Let a linear operator $R$ is defined by the formula  
\begin{equation}\label{linop}
R(q)=\alpha_{-1}q^{-1}+\alpha_{0}q^{0}+\alpha_{1}q^{1},
\end{equation}
where $q=q^{-1}+q^{0}+q^{1}$, $q^{-1}\in N$, $q^{0}\in \mathcal{G}_0$, $q^{1}\in
M$, $\alpha_{-1},\alpha_{0}, \alpha_{1} \in \R.$ Then the equation 
\begin{equation}\label{eq1}
q_{t}=[R(q),q], \quad q|_{t=0}=q_{0},
\end{equation}
can be reduced to a system of linear ODEs with variable coefficients  by the Golubchik-Sokolov construction (see Introduction).
\end{theorem}

\begin{proof} To apply the construction from \cite{golsok}, we take  for
$\bar{\mathcal{G}}$ the direct sum of three copies of the algebra $\mathcal{G}$, i.e. 
$$\bar{\mathcal{G}}=\mathcal{G} \oplus \mathcal{G} \oplus \mathcal{G}.$$

For $\bar {\mathcal{G}}_{+}$ we take the diagonal subalgebra in  $\bar{\mathcal{G}}$:
$$\bar{\mathcal{G}}_{+} = \{(a, a, a)|\, a \in \mathcal{G} \}, $$
and for  $\Bar{\mathcal{G}}_{-}$ the following subalgebra:
$$\bar{\mathcal{G}}_{-}= \{(a, b, c)|\, a \in \mathcal{G}_0+N, b\in \mathcal{G}_0+M, c\in \mathcal{G} \}.$$

Consider the vector subspace 
$$\bar{M}=\{(a,b,c)|\, a \in \mathcal{G}_0+N, b\in \mathcal{G}_0+M, c\in N+M \}.$$
Then condition (\ref{volchok7}) holds  for $\bar{\mathcal{G}}$ and $\bar{\mathcal{G}}_{+}, \bar{M}$,
i.e. $\bar{\mathcal{G}}=\bar{\mathcal{G}}_{+} \oplus \bar{M}$. Obviously, for the subalgebra $\bar{\mathcal{G}}_{+}$  condition (\ref{volchok8}) is fulfilled. It is easy to verify that for $\bar{M}$ we have 
$$[(\bar{\mathcal{G}}_{+}\cap \bar{\mathcal{G}}_{-}), \bar{M}] \subset \bar{M}. $$ 
Thus it follows from  \cite{golsok} that the equation
\begin{equation}\label{eq2}
\bar{q_{t}}=[\bar{q}_{+}, \bar{q}], \quad \bar{q}|_{t=0}=\bar{q_{0}},
\end{equation}
where $\bar{q} \in \bar{\mathcal{G}}, \, \bar{q}_{+}$ is the projection $\bar{q}$ onto $\bar{\mathcal{G}}_{+}$ parallel to  
$\bar{M}$, can be reduced to a system of linear ODEs with variable coefficients.

To complete the proof it is sufficient to choose $\bar{q}=(\alpha_{1} q, \alpha_{-1} q, \alpha_{0}q)$ where
$q \in \mathcal{G}$. Notice, that in this case $\bar{q}_{+}=(R(q), R(q), R(q))$ where the operator $R$ is defined by (\ref{linop}). Now one can see that for each component, equation  (\ref{eq2}) has the same form  $q_{t}=[R(q),q]$. This coincides with (\ref{eq1}).
\end{proof}

\begin{corollary} Let $\mathcal{G}=\bigoplus\limits_{i=-k}^{k} \mathcal{G}_{i}$ be a
$\Z$-graded Lie algebra and   $N=\bigoplus\limits_{i=-k}^{-1} \mathcal{G}_{i}$,
$M=\bigoplus\limits_{i=1}^{k} \mathcal{G}_{i}$.
Then all conditions of Theorem are fulfilled and the equation  
\begin{equation}\label{7}
q_{t}=\left[\alpha_{-1}\sum_{i=-k}^{-1}q_{i}+\alpha_{0}q_{0}+\alpha_{1}\sum_{i=1}^{k}
q_{i}, \; q\right] \end{equation} can be reduced to a system of linear differential equations with variable coefficients. 
\end{corollary}

\begin{remark}
It follows from (\ref{eq1}) that the traces of $q^i, \, i=1,2,... $ are polynomial first integrals for dynamical system (\ref{eq1}).  Since $q$ does not depend on any spectral parameter, the number of these integrals is not enough for the complete integrability of (\ref{eq1}). 
\end{remark}

\begin{remark} For an element $g \in \mathcal{G}$ such that $[g, \mathcal{G}_0] \subseteq  \mathcal{G}_0, [g, M] \subseteq  M, [g, N] \subseteq  N$ we have  $[g,R(q)]=R([g,q])$. Thus, $q_\tau =[g,q]$ is a linear symmetry for Equation~\eqref{eq1}. 
\end{remark}

\section{Examples}

Corollary 1 gives rise to the following two examples.

\begin{example} Consider the following graduation on $\mathcal{G}=\mathfrak{sl}_{2}$:

$$
\mathcal{G}_{-1}=\left\{\begin{pmatrix}
  0 & 0\\
  v & 0\\
\end{pmatrix}\right\}, \quad
\mathcal{G}_{0}=\left\{\begin{pmatrix}
  w & 0\\
  0 & -w\\
\end{pmatrix}\right\}, \quad
\mathcal{G}_{1}=\left\{\begin{pmatrix}
  0 & u\\
  0 & 0\\
\end{pmatrix}\right\}.$$
The equation \eqref{7} can be rewritten as the system of ODEs
\begin{equation}\label{sys0}
\begin{cases}
  w_t&=( \alpha_{1} - \alpha_{-1}) \,u v,\\
  u_t&= 2 (\alpha_{0} - \alpha_{1})\,u w, \\ 
  v_t&= 2 ( \alpha_{-1} - \alpha_{0} )\,v w. 
\end{cases}
\end{equation}
\noindent The system has the following linear infinitisemal symmetry
$$
\begin{cases}
  w_\tau &= 0,\\
  u_\tau&= u, \\
  v_\tau&= -v,
\end{cases}
$$
and possesses expected first integral (see Remark 1): 
$$
 H_1=\frac{1}{2} Tr(q^2)=u v+w^2.
$$
Note that the infinitesimal symmetry is generated by the the scalling symmetry group $u\to k u, $  $v\to k^{-1} v, $  $w\to w. $

System (\ref{sys0}) can be easily integrated in different ways. We shall demonstrate the Lie algorithm by this simple example. 
To apply it for an ODE system in $n$ variables, we need in total $n$ symmetries (including the initial system) and first integrals such that: 1) all symmetries commute each other; 2) all first integrals are first integrals for all symmetries as well. In our case we have two symmetries and one integral $H_1.$ It is easy to see that  
$\frac{d H_1}{d \tau}=0$.  Using the identity $H_1=C$, we eliminate $w$ and get

\begin{equation} \label{new}
\begin{cases}
  u_t&=  2 (\alpha_{0} - \alpha_{1})\, \,u \sqrt{C - u v }, \\
  v_t&= 2 ( \alpha_{-1} - \alpha_{0} ) \,v \sqrt{C - u v },
\end{cases}\qquad 
\begin{cases}
  u_\tau&= u, \\
  v_\tau&= -v.
\end{cases}
\end{equation}
Now we have to find a transformation  $\bar u=\varphi (u,v), \, \bar v=\psi (u,v)$ such that
\begin{gather*}
\begin{cases}
  \varphi_t&= 1, \\
  \psi_t&= 0.
\end{cases}\qquad
\begin{cases}
  \varphi_\tau &= 0, \\
  \psi_\tau &= 1.
\end{cases}
\end{gather*}
It is easy to see that the functions $\varphi, \psi$ satisfy the following overdetermined systems of PDEs

\begin{gather*}
\begin{cases}
 2 ((\alpha_{0} - \alpha_{1})\, \varphi_{u} u + ( \alpha_{-1} - \alpha_{0} )\,\varphi_{v} v) \sqrt{C - u v } = 1, \\
 \varphi_{u} u - \varphi_{v} v = 0,
\end{cases}\\
\begin{cases}
 2 ((\alpha_{0}-\alpha_{1})   \psi_{u} u + (\alpha_{-1} - \alpha_{0}) \psi_{v} v) \sqrt{C - u v }= 0, \\
 \psi_{u} u - \psi_{v} v = 1.
\end{cases}
\end{gather*}

From here we find all partial derivatives of $\varphi$ and $\psi$. The existence of the functions $\varphi$ and $\psi$ follows from the compatibility of $(\ref{new})$. Integrating, we obtain that
\begin{gather*}
\varphi (u,v) = \frac {\artanh{ \sqrt{1 - \frac{u v}{C}}}}{  (\alpha_{1}-\alpha_{-1})  \sqrt{C} }+ k_1, \\
\psi (u,v) = \frac {\alpha_{0}-\alpha_{-1}}{\alpha_{1}-\alpha_{-1}} \ln {(\alpha_{1}-\alpha_{-1}) u}+\frac {\alpha_{0}-\alpha_{1}}{\alpha_{1}-\alpha_{-1}} \ln {(\alpha_{1}-\alpha_{-1}) v}+k_2.
\end{gather*}

The functions $u(t)$ and $v(t)$ have to be found from $\psi(u,v)=t, \, \varphi(u,v)=0.$ Finally, we find
\begin{gather*}
u =  z k_3 \, \left(\cosh{z (\alpha_1-\alpha_{-1}) (t-k_1)}\right)^{\frac{2\, (\alpha_0-\alpha_{1})}{\alpha_1-\alpha_{-1}}}, \\
v = \frac{z}{k_3} \, \left(\cosh{z (\alpha_1-\alpha_{-1}) (t-k_1)}\right)^{\frac{2\, (\alpha_0-\alpha_{-1})}{\alpha_{-1}-\alpha_{1}}},\\
w = z \, \tanh{z (\alpha_1-\alpha_{-1}) (t-k_1)}, 
\end{gather*}
where $z, k_1, k_3$ are arbitrary constants. 
Here $z^2=C$ and the expression of $k_3$ in terms of $k_1,k_2,C$ is too large to be presented here.
From these formulas it follows that for general values of $\alpha_i$ the functions $u(t)$ and $v(t)$ have movable branching points in the complex plane.
\end{example}

Below we present not so trivial examples related to the Lie algebra $so(3,1)$. 
Integrable Hamiltonian dynamical systems with quadratic right hand sides associated with semi-simple Lie algebras have been considered in   \cite{miscfom, reysem, boma1}. The complex isomorphic algebras $so(4)$ and $so(3,1)$ are  interesting from a viewpoint of applications.   Integrable Hamiltonian systems associated with these algebras were studied in \cite{vesel, sokso4, sokwolf}.  
We present two new examples related to $so(3,1)$. 
Systems from these examples apparently are non-Hamiltonian.  
We find sufficiently many first integrals and infinitesimal symmetries to claim that these systems are locally integrable by the algorithm of S. Lie (see Example 1).

\begin{example} 
Let $\mathcal{G}=\{A\in \R_{4\times 4}|\,A^{*}= -A\},$ where the involution is defined by the formula  $A^{*}=TA^{t}T^{-1}$, $T=e_{11}+e_{22}+e_{34}+e_{43}$, the superscript $t$ means the matrix transposition, and   $e_{ij}$ denote the matrix unities.  If we take $T=e_{11}+e_{22}+e_{33}-e_{44}$ then the condition of skew-symmetricity with respect to the involution defines the algebra   $so(3,1).$ The isomorphism between  $\mathcal{G}$ and  $so(3,1)$ is given by  
$ \bar q=B^{-1} \,  q \, B$, where $q \in \mathcal{G},\, \bar q \in so(3,1)$, and the matrix $B$ equals
$$
 e_{11}+e_{22}+ \frac{\sqrt{2}}{2} (e_{33}+e_{34}+e_{43}-e_{44}).
$$ 

Define a 3-graduation on the algebra $\mathcal{G}$. Generic element of the algebra has the form  
$$
 q = \begin{pmatrix}
  0 & w_{1} &u_{1}&v_{1}\\
  -w_{1} & 0 &u_{2}&v_{2}\\
  -v_{1} & -v_{2} &w_{2}&0\\
  -u_{1} & -u_{2} &0&-w_{2}\\
\end{pmatrix}.
$$
It is easily verified that the decomposition of  $\mathcal{G}$ into the sum of the following components 
\begin{gather*}
 \mathcal{G}_{0} = \left\{ w_1 (e_{12}-e_{21})+w_2 (e_{33}-e_{44}) \right\},\\
 \mathcal{G}_{1} = \left\{ u_1 (e_{13}-e_{41})+u_2 (e_{23}-e_{42}) \right\},\\
 \mathcal{G}_{-1} = \left\{ v_1 (e_{14}-e_{31})+v_2 (e_{24}-e_{32}) \right\}
\end{gather*}
equips $\mathcal{G}$ with  a structure of a 3-graded Lie algebra.

It can be easily verified that equation \eqref{7} can be written as the system of ODEs
$$
\begin{cases}
  {w_1}_t&= (\alpha_{1} - \alpha_{-1}) \,(u_2 v_1 - u_1 v_2),\\
  {u_1}_t&= (\alpha_{0} - \alpha_{1}) \,(u_2 w_1 - u_1 w_2), \\
  {v_1}_t&= (\alpha_{0} - \alpha_{-1}) \,(v_2 w_1 + v_1 w_2), \\
  {w_2}_t&= (\alpha_{1} - \alpha_{-1}) \,(u_1 v_1 + u_2 v_2),\\
  {u_2}_t&= (\alpha_{1} - \alpha_{0}) \,(u_1 w_1 + u_2 w_2), \\
  {v_2}_t&= (\alpha_{-1} - \alpha_{0}) \,(v_1 w_1 - v_2 w_2).
\end{cases}$$

The system possesses the following linear and quadratic infinitesimal symmetries:  
\begin{gather*}
\begin{cases}
  {w_1}_{\tau_{1}} &= 0,\\
  {u_1}_{\tau_{1}} &= -u_1, \\
  {v_1}_{\tau_{1}} &= v_1 , \\
  {w_2}_{\tau_{1}} &= 0,\\
  {u_2}_{\tau_{1}} &= -u_2 , \\
  {v_2}_{\tau_{1}} &= v_2.
\end{cases} \qquad
\begin{cases}
  {w_1}_{\tau_{2}} &= 0,\\
  {u_1}_{\tau_{2}} &= u_2, \\
  {v_1}_{\tau_{2}} &= v_2 , \\
  {w_2}_{\tau_{2}} &= 0,\\
  {u_2}_{\tau_{2}} &= -u_1 , \\
  {v_2}_{\tau_{2}} &= -v_1.
\end{cases} \qquad
\begin{cases}
  {w_1}_{\tau_{3}} &= (- \alpha_{1} + \alpha_{-1}) \,(u_1 v_1 + u_2 v_2),\\
  {u_1}_{\tau_{3}} &= (- \alpha_{0} + \alpha_{1}) \,(u_1 w_1 + u_2 w_2), \\
  {v_1}_{\tau_{3}} &= (\alpha_{0} - \alpha_{-1}) \,(v_1 w_1 - v_2 w_2), \\
  {w_2}_{\tau_{3}} &= (- \alpha_{1} + \alpha_{-1}) \,(u_1 v_2 - u_2 v_1),\\
  {u_2}_{\tau_{3}} &= (\alpha_{1} - \alpha_{0}) \,(u_2 w_1 - u_1 w_2), \\
  {v_2}_{\tau_{3}} &= (\alpha_{0} - \alpha_{-1}) \,(v_2 w_1 + v_1 w_2).
\end{cases}
\end{gather*}

This system has also two polynomial first integrals of degree two: 
$$
 I_1=w_1 w_2 + u_1 v_2 - v_1 u_2
$$
and
$$
I_2=w_1^2+ 2 u_1 v_1 - w_2^2 + 2 u_2 v_2.
$$
The traces of powers of  $q$ are polynomials in these two integrals. It is easy to verify that all symmetries commute each 
other and that $I_1, I_2$ are the first integrals for all symmetries. So, the Lie algorithm is applicable. $\square$
\end{example}

In the next example the algebra $\mathcal{G}$ is not $\Z$-graded.
\begin{example} Let $\mathcal{G}$ be the same algebra as in Example 2. Now we consider
\begin{gather*}
\mathcal{G}_{0} = \left\{
 \begin{pmatrix}
 0 & w_1 & 0 & 0 \\
 -w_1 & 0 & 0 & 0 \\
 0 & 0 & k w_1 & 0 \\
 0 & 0 & 0 & -k w_1
 \end{pmatrix} \right\},\\
M = \left\{ \begin{pmatrix}
 0 & lw_2 & 0 & v_1 \\
 -lw_2 & 0 & 0 & v_2 \\
 -v_1 & -v_2 & w_2 & 0 \\
 0 & 0 & 0 & -w_2
 \end{pmatrix} \right\},\\
N = \left\{ \begin{pmatrix}
 0 & 0 & u_1 & 0 \\
 0 & 0 & u_2 & 0 \\
 0 & 0 & 0 & 0 \\
 -u_1 & -u_2 & 0 & 0
 \end{pmatrix} \right\},
\end{gather*}
where $k$ and $l$ are parameters such that $kl \ne 1$. It is easy to check that all conditions of the Theorem hold.
Equation \eqref{eq1} can be rewritten in the form:

$$
\begin{cases}
  {w_1}_t&= \dfrac{\alpha_{1} - \alpha_{-1}}{1-k l} \,(-l u_1 v_1 + u_2 v_1 - u_1 v_2 - l u_2 v_2),\\ 
  {u_1}_t&= (\alpha_{1} - \alpha_{0}) \, (k u_1 - u_2 ) w_1 + (\alpha_{1} - \alpha_{-1})  (u_1 - l u_2) w_2,\\
  {v_1}_t&= (\alpha_{0} - \alpha_{-1}) \,(k v_1 + v_2) w_1, \\ 
  {w_2}_t&= \dfrac{\alpha_{1} - \alpha_{-1}}{1-k l} \,(u_1 v_1 - k u_2 v_1 + k u_1 v_2 + u_2 v_2),\\ 
  {u_2}_t&= (\alpha_{1} - \alpha_{0}) \, (u_1 + k u_2 ) w_1 + (\alpha_{1} - \alpha_{-1})  (l u_1 +  u_2) w_2,\\ 
  {v_2}_t&= (\alpha_{-1} - \alpha_{0}) \,(v_1 - k v_2) w_1.
\end{cases}$$

This system possesses the same linear infinitesimal symmetries as the system in Example~2.
It also has the following quadratic symmetry:
$$
\begin{cases}
{w_1}_\tau &= -\dfrac{(l^2+1)(u_1 v_1 - k u_2 v_1 + k u_1 v_2 + u_2 v_2)}{kl - 1},\\
\\
{u_1}_\tau &= (k^2+1)(l u_2 w_1 -  u_1 w_1) - \\
& - \dfrac{k \,(\alpha_0 + \alpha_1 - 2 \alpha_{-1}) + 2 l\,(\alpha_1 - \alpha_{-1} )+  k l^2\,(\alpha_0  - \alpha_1) }{\alpha_1 - \alpha_{-1}} u_1 w_2 + \\
& + \dfrac{\alpha_0 - \alpha_1 + 2 k l\,(\alpha_1 - \alpha_{-1} ) + l^2\,(\alpha_0  + \alpha_1 - 2 \alpha_{-1})}{\alpha_1 - \alpha_{-1}} u_2 w_2,\\
\\
{v_1}_\tau &= -\dfrac{(\alpha_{-1} - \alpha_0) (l^2+1) (k v_1 + v_2) w_2}{\alpha_1 - \alpha_{-1}},\\
\\
{w_2}_\tau &= \dfrac{k^2 + 2 k l - 1}{k l - 1} (u_1 v_2 - u_2 v_1)
        - \dfrac{k^2 l - 2k - l}{k l - 1} (u_1 v_1 + u_2 v_2),\\
\\        
 {u_2}_\tau &= -(k^2 + 1) (l u_1 w_1 + u_2 w_1) - \\
         &-\dfrac{\alpha_0 - \alpha_1 + 2 k l\,(\alpha_1  - \alpha_{-1}) +  l^2\,(\alpha_0 + \alpha_1 - 2 \alpha_{-1})}{\alpha_1 - \alpha_{-1}} u_1 w_2 - \\
        & -\dfrac{k (\alpha_0 + \alpha_1 - 2 \alpha_{-1}) + 2 l\,(\alpha_1 - \alpha_{-1}) + k l^2\,(\alpha_0  - \alpha_1)}{\alpha_1 - \alpha_{-1}} u_2 w_2,\\
\\
{v_2}_\tau &= \dfrac{(\alpha_{-1} - \alpha_0) (l^2 + 1) (v_1 - k v_2) w_2}{\alpha_1 - \alpha_{-1}}.
\end{cases}
$$

Moreover, the system has the following quadratic first integrals: 
\begin{gather*}
 I_1=-2 u_1 v_1 - 2 u_2 v_2 +  (k w_1 + w_2)^2 -  (w_1 + l w_2)^2, \\[2mm] 
\qquad I_2 = u_2 v_1 - u_1 v_2 - (k w_1 + w_2) (w_1 + l w_2).
\end{gather*}
\noindent Notice that $I_1$ is the trace of $q^2$. It is easy to verify that the Lie algorithm is applicable. $\square$
 
\end{example}

\begin{remark} Note that some first integrals for the systems from Examples 1-3   are not single valued in the complex plane. For instance, $$
 H_2= u^{\alpha_{0}-\alpha_{-1}} v^{\alpha_{0}-\alpha_{1}}
$$
is a first integral for (\ref{sys0}). 
 The system from Example 2 possesses first integrals of the form $I=f(u_1,u_2,v_1,v_2).$ It is easy to check that such integrals should satisfy two equations of the form $X(f)=0$ and $Y(f)=0,$ where $X,Y$ are given by:
\begin{gather*}
X=(\alpha_{1} - \alpha_{0})\, u_1 \frac{\partial}{\partial u_{2}}+(\alpha_{-1} - \alpha_{0})\, v_1 \frac{\partial}{\partial v_{2}} + (\alpha_{0} - \alpha_{1}) \, u_2 \frac{\partial}{\partial u_{1}} + (\alpha_{0} - \alpha_{-1}) \, v_2 \frac{\partial}{\partial v_{1}},\\
Y=(\alpha_{1} - \alpha_{0})\, u_1 \frac{\partial}{\partial u_{1}} + (\alpha_{0} - \alpha_{-1}) \, v_1 \frac{\partial}{\partial v_{1}}+(\alpha_{1} - \alpha_{0}) \,u_2 \frac{\partial}{\partial u_{2}}+(\alpha_{0} - \alpha_{-1}) \, v_2 \frac{\partial}{\partial v_{2}}.
\end{gather*}
Solving this system of two PDEs, we find it's simplest solutions  
\begin{gather*}
 I_3=\tg\left(\frac{\alpha_{-1} - \alpha_{0}}{\alpha_{1} - \alpha_{0}} \arctg\frac{u_1}{ u_2}-\arctg{\frac{v_1}{ v_2}}\right),\\
 I_4=(u_1^2+u_2^2)^{\frac{\alpha_{-1} - \alpha_{0}}{\alpha_{1} - \alpha_{0}}} (v_1^2+v_2^2) \sin^2\left(\frac{\alpha_{-1} - \alpha_{0}}{\alpha_{1} - \alpha_{0}} \arctg\frac{u_2}{ u_1}-\arctg{\frac{v_2}{ v_1}}\right).
\end{gather*}
\end{remark}

It is remarkable that  the number of polynomial symmetries and first integrals in Examples 1-3 is equal to the number of independent variables and we may apply the Lie algorithm to integrate the systems in quadratures. We don't know whether it is true for all systems described in Theorem of Section 2.  

\section*{Acknowledgements.}
The authors are deeply grateful to I.~Z.~Golubchik and V.~V.~Sokolov for useful discussions and constant attention to this work, and to A.~I.~Zobnin for the attention and the assistance in working with text.  The research of O.~S. was partially supported by the RFBR grant 11-01-00341-a.

\end{document}